\documentclass[a4paper,10pt]{article}
\usepackage{graphicx}
\usepackage{color}\usepackage{amsmath, amsthm, slashed, enumerate}
\usepackage{amssymb}
\usepackage{rotating}
\usepackage{epsfig}
\usepackage{verbatim}
\usepackage{rotating}
\usepackage{graphicx}
\newtheorem{definition}{Definition}[section]
\newtheorem{theorem}{Theorem}[section]

\newtheorem{corollary}{Corollary}[section]
\newtheorem{proposition}{Proposition}[section]
\newtheorem{lemma}{Lemma}[section]
\newtheorem{remark}{Remark}[section]

\newtheorem{example}{Example}[section]

\title{Well-posedness for the massive wave equation on asymptotically anti-de Sitter spacetimes}
\author{Gustav Holzegel\thanks{Princeton University,
Department of Mathematics, Fine Hall, Washington Road,
Princeton, NJ 08544 United States}}

\begin{document}
\maketitle
\begin{abstract}
{\bfseries Abstract.}\quad In this paper, we prove a well-posedness theorem for the massive wave equation (with the mass satisfying the Breitenlohner-Freedman bound) on asymptotically anti-de Sitter spaces. The solution is constructed as a limit of solutions to an initial boundary value problem with boundary at a finite location in spacetime by finally pushing the boundary out to infinity. The solution obtained is unique within the energy class (but non-unique if the decay at infinity is weakened).
\end{abstract}


\section{Introduction}
In \cite{HolzegelAdS}, the author obtained uniform global bounds for a class of solutions to the massive wave equation
\begin{equation} \label{mwe}
\Box_g \psi - \alpha \frac{\Lambda}{3} \psi = 0 \, 
\end{equation}
on any black hole spacetime suitably close to a slowly rotating Kerr-AdS solution provided the spacetime admitted a Killing field on the black hole exterior which is everywhere causal and null on the horizon. The more elementary issue of well-posedness of (\ref{mwe}) on these backgrounds was not discussed in the aforementioned paper. Instead, we assumed the existence of solutions to (\ref{mwe}) arising from suitably regular initial data and exhibiting a sufficiently strong decay at null-infinity for our argument to work. The main purpose of this paper is to prove the existence of such solutions, indeed, the well-posedness of a suitable initial value problem. 

We direct the reader to section 8.4 of \cite{Mihalisnotes} and to the previous paper \cite{HolzegelAdS} for further background. We also refer to the recent work of Vasy \cite{Vasy2} for related local and global well-posedness results on asymptotically AdS spacetimes (including a detailed description of the propagation of singularities)  using in particular tools from microlocal analysis.

The reason that even the well-posedness question is non-trivial arises from the fact that asymptotically AdS spacetimes are not globally hyperbolic. To make hyperbolic equations like (\ref{mwe}) well-posed on such backgrounds, in general, boundary conditions will have to be imposed on the timelike boundary at infinity, commonly referred to as ``scri" and denoted $\mathcal{I}$. It turns out that existence and uniqueness is particularly subtle in the range of ``negative mass" $0<\alpha<\frac{9}{4}$, as was first observed in the pioneering work of Breitenlohner and Freedman \cite{Breitenlohner} by a mode analysis on exactly AdS spacetimes. For this case, Bachelot \cite{Bachelot} finally proved that, in the range $\frac{5}{4}\leq \alpha<\frac{9}{4}$, infinitely many solutions of (\ref{mwe}) exist, depending on boundary conditions, while for $\alpha < \frac{5}{4}$, the problem is well-posed, provided the spherically symmetric mode is subtracted. Since this rather surprising behavior is a consequence of the timelike boundary at infinity, one may expect to establish a well-posedness statement not only for pure AdS but for all asymptotically AdS spacetimes. This is carried out in the present paper. 

While the work of \cite{Bachelot} studies the problem in the realm of scattering theory and self-adjoint extensions of operators, the approach taken here is entirely based on energy estimates.
In this context, we prove a well-posedness statement for the complete range $\alpha<\frac{9}{4}$, with the solution being unique in the energy class. Key is the observation that the non-uniqueness revealed in \cite{Bachelot} and \cite{Breitenlohner} stems from solutions which do not satisfy the energy estimate, or more precisely, solutions whose $\partial_t$-energy-flux through $\mathcal{I}$ is infinite. In other words, the solution we construct for the given initial boundary value problem is unique in a class of solutions which decay sufficiently rapidly at infinity.

The actual construction of the solutions is carried out in what we will call \emph{an asymptotically AdS patch} near infinity. This is sufficient because the well-posedness statement for the boundary initial value problem in the patch can then be combined with elementary global considerations involving domain of dependence arguments (cf.~section \ref{gloco}) to establish a spatially global well-posedness statement for a suitable class of spacetimes.

With such an asymptotically AdS patch being fixed, note that the boundary initial value problem with the boundary located somewhere \emph{inside} the spacetime patch can be solved using standard techniques, as the weights in $r$ ($r$ being a radial coordinate) remain bounded everywhere. With this in mind, we construct a sequence of timelike boundaries $\mathcal{B}_i$ which approach the timelike hypersurface $\mathcal{I}$ in the limit $i \rightarrow \infty$ and consider the sequence of solutions $\mathcal{S}_i$ associated with each initial boundary value problem. Now the crucial observation is that actually, stronger $r$-weighted norms than the ones arising from the energy are propagated by the equation. In this way, we establish improved uniform (in $i$) $r$-weighted estimates for each solution $\mathcal{S}_i$. Finally, we compare two such solutions in the region where they are both defined and establish convergence in the energy norm using the improved uniform estimates mentioned above.

Here is an outline of the paper. In section \ref{aAdS} we define the notion of an asymptotically AdS patch near infinity (which, in particular, imposes suitable decay assumptions on the metric) which provides the arena in which we are going to prove the well-posedness statement. Section \ref{norms} defines various $r$-weighted Sobolev norms that we are going to use. Certain elliptic estimates on spacelike slices, which will be used later in the argument, are derived in section \ref{elliptic}.  After defining the class of initial data in section \ref{idata}, the main theorem is stated in section \ref{theosec}. Its proof is reduced to a key proposition, Proposition \ref{mnthm}, which we prove in section \ref{theproof} using the limiting procedure outlined above. In the final section \ref{finsec}, we globalize our well-posedness result in the asymptotically AdS patch to a spatially global statement for a suitable class of spacetimes (cf.~section \ref{gloco}). As an afterthought we also discuss the regularity of the metric required in the case of spherical symmetry. This will be important for the applications of this paper to the non-linear setting of \cite{gs:lwp}.

\section{Asymptotically AdS spacetimes} \label{aAdS}
From the extensive literature on asymptotically AdS spacetimes (cf.~for instance \cite{Henneaux, Hollands2}), we are going to distill the following definition, which is most useful for our purposes. 
It is to be thought of as defining a patch of spacetime near the timelike boundary at infinity.
\begin{definition}
We call the Lorentzian manifold $\left(\mathcal{\tilde{D}},g\right)$ an asymptotically AdS spacetime patch with cosmological constant $\Lambda = -\frac{3}{l^2}$ if the following holds:
\begin{enumerate}[(I)]
\item $\mathcal{\tilde{D}}$ is topologically $\left[0,T\right] \times [\tilde{R},\infty) \times S^2$. It is covered by a coordinate system $(t,r,x_i,y_i)$ ($i=1,2$ being two coordinate patches covering the spheres) such that hypersurfaces of constant $t$ foliate $\mathcal{\tilde{D}}$ and are spacelike. We denote them by $\tilde{\Sigma}_t$. Moreover, any surface $\tilde{\Sigma}_t$ is itself foliated by $2$-spheres $S^2_{t,r}$. The boundary of the region $\mathcal{\tilde{D}}$ is given by the slice $\tilde{\Sigma}_0$ (in the past), the slice $\tilde{\Sigma}_{T}$ (in the future) and the timelike hypersurface $\tilde{\mathcal{B}}_0$, which is generated by the integral curves of the vectorfield $\partial_t$ emanating from the sphere $S^2_{0,\tilde{R}}$. Finally, all hypersurfaces of constant $r$ are timelike in $\mathcal{\tilde{D}}$.

\item In the coordinates $(t,r,x_i,y_i)$, the metric $g$, which is assumed to be smooth, has the following asymptotic behavior ($A,B \in \{x_i,y_i\}$):
\begin{align}
g_{tt} = - \frac{r^2}{l^2} - 1 + \mathcal{O}\left(\frac{1}{r}\right) \textrm{ \ \ \ \ , \ \ \ \ } g_{rr} = \frac{l^2}{r^2} -\frac{l^4}{r^4} + \mathcal{O}\left(\frac{1}{r^5}\right) \nonumber \\
g_{AB} = r^2 \left[\slashed{g}_{S^2}\right]_{AB} + \mathcal{O}\left(\frac{1}{r}\right) \phantom{XXXXXX} \nonumber \\
g_{tr} = \mathcal{O}\left(\frac{1}{r^3}\right)  \textrm{ \ \ \ , \ \ \ }
g_{tA} = \mathcal{O}\left(\frac{1}{r}\right)  \textrm{ \ \ \ , \ \ \ } g_{rA} = \mathcal{O}\left(\frac{1}{r^4}\right) 
\end{align}
where $\slashed{g}_{S^2}$ denotes the standard metric on the round 2-sphere. For the inverse
\begin{align}
g^{tt} = \frac{l^2}{r^2} - \frac{l^4}{r^4} +  \mathcal{O}\left(\frac{1}{r^5}\right) \textrm{ \ \ \ \ , \ \ \ \ } g^{rr} = \frac{r^2}{l^2} + 1 + \mathcal{O}\left(\frac{1}{r}\right) \nonumber \\
g^{AB} = \frac{1}{r^2} \left[\slashed{g}_{S^2}\right]^{AB} + \mathcal{O}\left(\frac{1}{r^5}\right) \phantom{XXXXXX} \nonumber \\
g^{tr} =  \mathcal{O}\left(\frac{1}{r^3}\right) \textrm{ \ \ \ , \ \ \ }
g^{tA} =   \mathcal{O}\left(\frac{1}{r^5}\right) \textrm{ \ \ \ , \ \ \ }
g^{rA} =  \mathcal{O}\left(\frac{1}{r^4}\right) \nonumber \, .
\end{align}
We assume similar decay for derivatives of the metric functions with each $r$-derivative lowering the powers of $r$ by one. For instance,
\begin{equation}
\partial_r g_{tt} = -2\frac{r}{l^2} + \mathcal{O}\left(\frac{1}{r^2}\right) \textrm{ \ \ \ and \ \ \ } \partial_t g_{tt} = \mathcal{O}\left(\frac{1}{r}\right) \, .
\end{equation}
\end{enumerate}
\end{definition}
\begin{remark} \label{expl}
Note that the Kerr-AdS metric in the standard Boyer-Lindquist coordinates does not obey the fall-off imposed by (II). However, one can do a coordinate transformation \cite{Henneaux} leading to a system which admits the decay stated above. For the purpose of the existence proof below, one can in fact work with much weaker decay but we will not spell out the minimal assumptions here.
\end{remark}
\begin{remark}
With these assumptions on the metric one shows that there exist future complete outgoing null-geodesics in $\mathcal{\tilde{D}}$ which emanate from $\tilde{\Sigma}_0$. The limit endpoints of such geodesics are not in $\mathcal{\tilde{D}}$ but they can, at least formally, be parametrized as $\left(t,r=\infty,x_i,y_i\right)$, an ideal timelike boundary commonly referred to as $\mathcal{I}$. 
\end{remark}

For future reference, we also collect the estimates for the spactime volume form and the volume form induced on the slices $\tilde{\Sigma}_t$. This can be computed from (II) above. For convenience, we choose the coordinates $x_i, y_i$ on the $S^2_{t,r}$ such that $r^2 \sqrt{\slashed{g}_{S^2}} =r^2$ holds if $S^2_{t,r}$ is equipped with the round metric $r^2 \slashed{g}_{S^2}$. With this choice we have
\begin{equation}
\Big|\frac{\sqrt{g}}{r^2} - 1 \Big| \leq \frac{C}{r^3} \textrm{ \ \ \ \ and \ \ \ \ }
\Big| \frac{\sqrt{g_{\tilde{\Sigma}_t}}}{l r} -1 \Big| \leq \frac{C}{r^2} 
\end{equation}
respectively (while in general, ``$1$" is to be replaced by $\sqrt{\slashed{g}_{S^2}}$). Here and in the following, $C$ denotes a uniform constant depending only on the background $(\mathcal{\tilde{D}},g)$, which we will from now on regard as being fixed.

\subsection{The boundaries $\mathcal{B}_i$} \label{bndsec}
Let $r_i=2^i \cdot R$ for $i$ a non-negative integer and $R\geq \tilde{R}$ a large constant. 
Let us define the restrictions
\begin{align} \label{reR}
\mathcal{D} = {\mathcal{\tilde{D}}} \cap \{ r\geq R \} \ \ \ \ \ \ , \ \ \ \ \ \Sigma_t = \tilde{\Sigma}_t \cap \{ r \geq R \} \, .
\end{align}
The domain $\mathcal{D}$ (for a suitable choice of $R$, which will me made below and depend only on the fixed asymptotically AdS spacetime patch $(\mathcal{\tilde{D}},g)$) is the domain in which we are going to prove a well-posedness statement below.

Consider next the integral curves of the vectorfield $\partial_t$ emanating from $\Sigma_0 \cap \{r=r_i\}$: They generate a three-dimensional timelike hypersurface of constant $r=r_i$, which we will call $\mathcal{B}_{i}$. As $r_i \rightarrow \infty$, the ideal boundary $\mathcal{I}$ is approached. 
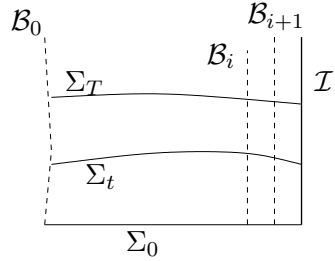
\[
\begin{picture}(0,0)%
\includegraphics{asympAdS3.pstex}%
\end{picture}%
\setlength{\unitlength}{2210sp}%
\begingroup\makeatletter\ifx\SetFigFont\undefined%
\gdef\SetFigFont#1#2#3#4#5{%
  \reset@font\fontsize{#1}{#2pt}%
  \fontfamily{#3}\fontseries{#4}\fontshape{#5}%
  \selectfont}%
\fi\endgroup%
\begin{picture}(3375,2863)(1276,-5546)
\put(2551,-5461){\makebox(0,0)[lb]{\smash{{\SetFigFont{11}{13.2}{\rmdefault}{\mddefault}{\updefault}{\color[rgb]{0,0,0}$\Sigma_0$}%
}}}}
\put(2101,-4711){\makebox(0,0)[lb]{\smash{{\SetFigFont{11}{13.2}{\rmdefault}{\mddefault}{\updefault}{\color[rgb]{0,0,0}$\Sigma_{t}$}%
}}}}
\put(3901,-2911){\makebox(0,0)[lb]{\smash{{\SetFigFont{11}{13.2}{\rmdefault}{\mddefault}{\updefault}{\color[rgb]{0,0,0}$\mathcal{B}_{i+1}$}%
}}}}
\put(3451,-3361){\makebox(0,0)[lb]{\smash{{\SetFigFont{11}{13.2}{\rmdefault}{\mddefault}{\updefault}{\color[rgb]{0,0,0}$\mathcal{B}_i$}%
}}}}
\put(4651,-3661){\makebox(0,0)[lb]{\smash{{\SetFigFont{11}{13.2}{\rmdefault}{\mddefault}{\updefault}{\color[rgb]{0,0,0}$\mathcal{I}$}%
}}}}
\put(1276,-2986){\makebox(0,0)[lb]{\smash{{\SetFigFont{11}{13.2}{\rmdefault}{\mddefault}{\updefault}{\color[rgb]{0,0,0}$\mathcal{B}_{0}$}%
}}}}
\put(1876,-3661){\makebox(0,0)[lb]{\smash{{\SetFigFont{11}{13.2}{\rmdefault}{\mddefault}{\updefault}{\color[rgb]{0,0,0}$\Sigma_{T}$}%
}}}}
\end{picture}%

\]
A computation reveals that 
\begin{itemize}
\item the spacelike hypersurfaces $\Sigma_t$ have normal
\begin{equation}
n_{\Sigma_t} = \left(\frac{l}{r} + a\left(t,r,x,y\right) \right)\partial_t + b\left(t,r,x,y\right) \partial_r + c^A\left(t,r,x,y\right) \partial_A
\end{equation}
with the estimate $r^3 |a| + r^2 |b| + r^4|c^A| \leq C$.
\item the timelike hypersurfaces $\mathcal{B}_i$ have normal
\begin{equation}
n_{\mathcal{B}_i} = \left(\frac{r}{l} + \tilde{b} \left(t,r,x,y\right) \right)\partial_r + \tilde{a}\left(t,r,x,y\right) \partial_t  + \tilde{c}^A\left(t,r,x,y\right) \partial_A
\end{equation}
with the estimate $r^4 |\tilde{a}| + r|\tilde{b}| + r^5|\tilde{c}^A| \leq C$ . 
\end{itemize}

We denote by $\mathcal{D}_i$ the subset of $\mathcal{D}$ that lies to the the inside of the cylinder $\mathcal{B}_i$, i.e~$\mathcal{D}_i = \mathcal{D} \cap \{ r < r_i\}$. Finally, $\Sigma_t^i = \Sigma_t \cap \mathcal{D}_i$.

\subsection{Approximate Killing fields near infinity} \label{approK}
We have the following formula for the deformation tensor of a vectorfield $X$:
\begin{equation}
 2\phantom{}^{(X)}\pi^{ab} = g^{ac} \partial_c X^b + g^{bd} \partial_d X^a + g^{ac} g^{bd} g_{cd,f} X^f \, .
\end{equation}
Using the decay of the metric components near infinity one deduces
that the vectorfield $\partial_t$ is approximately Killing near infinity, in the sense that its
deformation tensor satisfies
\begin{equation}
|\phantom{}^{(\partial_t)}\pi^{tt}|r^5 + |\phantom{}^{(\partial_t)}\pi^{tr}|r^3 + |\phantom{}^{(\partial_t)}\pi^{tA}|r^5 + |\phantom{}^{(\partial_t)}\pi^{rr}|r + |\phantom{}^{(\partial_t)}\pi^{rA}|r^4 + |\phantom{}^{(\partial_t)}\pi^{AB}|r^5 \leq C \, .
\end{equation}

\section{The norms} \label{norms}
For any real number $s$ and $n=0,1,2$ we define the weighted Sobolev spaces 
\begin{align}
H_{AdS}^{n,s}\left(\Sigma\right) := \Big\{ \psi \in H^n_{loc} \left(\Sigma\right) \ \ , \ \  \|\psi\|^2_{H_{AdS}^{n,s}\left(\Sigma\right)}  < \infty \Big\}
\end{align}
with the norms defined as follows:
\begin{equation}
\|\psi\|^2_{H_{AdS}^{0,s}\left(\Sigma\right)} = \int_{\Sigma} r^s \psi^2 r^{2} dr \, d\omega \, ,
\end{equation}
\begin{equation}
\|\psi\|^2_{H^{1,s}_{AdS}\left(\Sigma\right)} = \int_{\Sigma} r^s \left[r^2 \left(\partial_r \psi\right)^2 + |\slashed{\nabla}\psi|^2 + \psi^2 \right] r^{2} dr  \, d\omega \, ,
\end{equation}

\begin{eqnarray}
\|\psi\|^2_{H_{AdS}^{2,s}\left(\Sigma\right)} = \int_{\Sigma} r^s \Bigg[r^4 \left(\partial_r \partial_r \psi\right)^2 + r^2 | \slashed{\nabla}\partial_ r\psi|^2 + |\slashed{\nabla} \slashed{\nabla}\psi|^2 \nonumber \\ + r^2 \left(\partial_r \psi\right)^2 + |\slashed{\nabla}\psi|^2 + \psi^2\Bigg] r^{2} dr  \, d\omega \, .
\end{eqnarray}
where $d\omega=dx dy$, and $ |\slashed{\nabla}\psi|^2$ denotes the norm induced on the $S^2_{t,r}$ with $\slashed{\nabla}$ being the gradient of the induced metric. Note that in view of the closeness assumptions on the metric, the angular norm using the actual induced metric is equivalent to the one using the round metric.
\begin{remark}
The reader should note that, at least for $s=0$, these norms have their origin in the standard energy estimate arising from the approximate Killing field $\partial_t$, cf.~(\ref{emt}) and (\ref{emtid}). For instance, one easily sees
\begin{equation}
 \int_{\Sigma_t}  \mathbb{T}\left[\psi\right] \left(\partial_t, n_{\Sigma_t}\right) \sqrt{g_{\Sigma_t}} dr \, d\omega \leq C  \left[ \| \psi \|_{H^{1,0}_{AdS}\left(\Sigma\right)}  + \| \partial_t \psi \|_{H_{AdS}^{0,-2}\left(\Sigma\right)}\right]  \, .
\end{equation}
The additional $r$-weight is motivated by certain weighted elliptic estimates to be proven later.
\end{remark}
\section{Some elliptic estimates} \label{elliptic}
An important element of the proof will be (weighted) elliptic estimates derived from the elliptic part of the wave operator on the spacelike slices $\Sigma_\tau$. Such estimates will allow us to essentially estimate \emph{all} derivatives from estimates on the time derivatives. The difficulty here lies in the understanding of the admissible $r$-weights.

\begin{lemma} \label{dens}
The space of smooth functions with compact support in $\Sigma$ is dense in $H_{AdS}^{m,s}\left(\Sigma\right)$.
\end{lemma}
\begin{proof}
Let $\phi \in H^{m,s}_{AdS}\left(\Sigma\right)$ and $\epsilon>0$ be given. Recall that $r_i = 2^i R$ and let $\chi_i\left(r\right)$ be a smooth bounded non-negative function, which is equal to $1$ for $r \leq 2^i R$ and equal to zero on $[2^{i+1}R,\infty)$. It is easy to see that $\chi_n$ can be chosen such that $|\partial^m_r \chi_i| \leq \frac{C_m}{r^m}$ for any non-negative integer $m$: Simply smooth the piecewise linear function
\begin{equation*} 
\chi^{lin}_{i} \left(r\right) =
\begin{cases} 1 & \text{if } r < 2^i R,\\ 2-\frac{r}{2^iR} & \text{if } 2^iR \leq r \leq 2^{i+1}R ,\\ 0 & \text{if } r > 2^{i+1}R.
\end{cases} \end{equation*}

 Define $\tilde{\phi}_i = \chi_i\left(r\right) \cdot \phi $. In view of the condition on the derivatives of $\chi_i$, $\tilde{\phi}_i \in H^{m,s}_{AdS}\left(\Sigma\right)$ and moreover, $\tilde{\phi}_i$ is of compact support. In addition,
\begin{align}
\| \tilde{\phi}_i - \phi\|_{H_{AdS}^{m,s}\left(\Sigma\right)}  = \| \tilde{\phi}_i - \phi \|_{H_{AdS}^{m,s}\left(\Sigma \cap \{2^i R < r < 2^{i+1}R \}\right)} + \| \phi \|_{H_{AdS}^{m,s}\left(\Sigma \cap \{r \geq 2^{i+1}R \}\right)}\nonumber \\  \leq   \| \tilde{\phi}_i  \|_{H_{AdS}^{m,s}\left(\Sigma \cap \{2^i R < r < 2^{i+1}R \}\right)} + \| \phi \|_{H_{AdS}^{m,s}\left(\Sigma \cap \{r \geq 2^iR  \}\right)} \nonumber \\
\leq C_m  \| \phi \|_{H_{AdS}^{m,s}\left(\Sigma \cap \{r \geq 2^iR  \}\right)} < \frac{\epsilon}{2} \nonumber
\end{align}
for $i$ sufficiently large.\footnote{The series $\sum_{i=0}^\infty \| {\phi}  \|^2_{H_{AdS}^{m,s}\left(\Sigma \cap \{2^i R < r < 2^{i+1}R \}\right)}$ converges, since ${\phi} \in H_{AdS}^{m,s}\left(\Sigma\right)$. This implies that the sequence of partial sums $\sum_{i=0}^N \| {\phi}  \|^2_{H_{AdS}^{m,s}\left(\Sigma \cap \{2^i R < r < 2^{i+1}R \}\right)}$ is Cauchy. Thus for any $\epsilon>0$ there exists an $N$, such that $\sum_{i=N}^{N+k} \| {\phi}  \|^2_{H_{AdS}^{m,s}\left(\Sigma \cap \{2^i R < r < 2^{i+1}R \}\right)}<\epsilon$ for all $k$. \label{ftnc}}
This shows that $\phi \in H^{m,s}_{AdS}\left(\Sigma\right)$ can be approximated by functions of compact support in $H^{m,s}_{AdS}\left(\Sigma\right)$. By a standard mollifying argument we can smooth $\tilde{\phi}$ and hence approximate within $\epsilon$ in the space of smooth functions of compact support.
\end{proof}

Recall that $C$ is a uniform constant depending only on the fixed 
asymptotically AdS spacetime.
\begin{lemma} \label{fila}
Let $s\geq 0$. For any $\psi \in H^{1,s}_{AdS}\left(\Sigma_{\tau}\right)$ we have the Hardy inequality
\begin{align} \label{filaeq}
\frac{1}{l^2} \int_{\Sigma_\tau} \psi^2 r^s \ \sqrt{g} dr d\omega + c \int_{\Sigma_{\tau}} r^{2+s} \left(\partial_r \psi\right)^2 \ dr d\omega \nonumber \\
 \leq \frac{4}{\left(3+s\right)^2} \int_{\Sigma_\tau} \sqrt{g} g^{rr} \left(\partial_r\psi\right)^2 r^s dr d\omega \, 
\end{align}
provided that $R$ (defining $\Sigma_\tau$ in (\ref{reR})) is chosen sufficiently large depending only on the fixed asymptotically AdS spacetime. Here $c>0$ is a uniform constant.
\end{lemma}
\begin{proof}
Since the functions of compact support are dense in $H^{1,s}_{AdS}\left(\Sigma_{\tau}\right)$, it suffices to establish the estimate for this class. In that case, there is no boundary term at infinity in the integration by parts. Hence we find using Cauchy-Schwarz (and the fact that the boundary term at $r=R$ has a good sign)
\begin{align}
\int_{\Sigma_\tau} \psi^2 r^s \sqrt{g} dr d\omega \leq \frac{4}{\left(3+s\right)^2} \int_{\Sigma_\tau} \left[\frac{\left(3+s\right)^2\mathcal{G}_s^2}{\sqrt{g} r^s}\right] \left(\partial_r\psi\right)^2 dr d\omega \, ,
\end{align}
where $\mathcal{G}_s$ satisfies $\partial_r \mathcal{G}_s = r^s \sqrt{g}$. The Lemma now follows by observing that with the asymptotics of the metric we have $|\mathcal{G}_s| \leq \frac{r^{3+s}}{3+s} + C r^s \ln r$ and hence
\begin{align} 
\left(3+s\right)^2 \mathcal{G}_s^2 \leq r^{6+2s} + C r^{3+2s} \ln r 
\leq l^2 \cdot r^{2s} \cdot r^4 \left(\frac{r^2}{l^2} + 1  \right) - \frac{1}{2} l^2 r^{4+2s} 
\nonumber \\
\leq l^2 r^{2s}  \left(\left(\sqrt{g}\right)^2 + \mathcal{O}\left(r\right)\right)  \left(g^{rr} + \mathcal{O}\left(r^{-1}\right) \right)  - \frac{1}{2} l^2 r^{4+2s}
\nonumber \\
 \leq  l^2 \cdot g^{rr} \left(\sqrt{g}\right)^2 r^{2s} - \frac{1}{4} l^2 \cdot r^{4+2s} 
\end{align}
in $\mathcal{D}$, provided $R$ is chosen sufficiently large.
\end{proof}
The next proposition regards the solutions to an elliptic equation. The operator is the elliptic part of the wave operator:
\begin{align} \label{Ldef}
L \psi := \frac{1}{\sqrt{g}}\partial_i \left(g^{ij} \sqrt{g}\partial_j \psi\right) + \frac{1}{\sqrt{g}} \partial_t \left(g^{tj} \sqrt{g}\right) \partial_j \psi + \frac{\alpha}{l^2} \psi 
\end{align} 
with Roman indices ranging over $r,x,y$ and $\alpha<\frac{9}{4}$.
\begin{proposition} \label{sela}
Let $\mathcal{F} \in H^{0,s}_{AdS} \left(\Sigma_\tau\right)$ for all $s<\sqrt{9-4\alpha}$.
On any spacelike slice $\Sigma_{\tau} \subset \mathcal{D}$, the elliptic problem 
\begin{equation} \label{ellP}
(ELP) = \left\{
\begin{array}{rl} L \psi = \mathcal{F} & \textrm{in } \Sigma_{\tau} \\
\psi=0 &  \textrm{for} \ \  r \rightarrow \infty \\
\psi=0 & \textrm{at} \ \ r=R
\end{array} \right.
\end{equation}
has a unique solution in $H^{2,s}_{AdS} \left(\Sigma_{\tau}\right)$, $0\leq s<\sqrt{9-4\alpha}$ satisfying the estimate
\begin{align}
\| \psi \|^2_{H^{2,s}_{AdS} \left(\Sigma_{\tau}\right)} \leq C_s \int_{\Sigma_{\tau}} \mathcal{F}^2 r^{2+s} dr d\omega \, 
\end{align}
provided that $R$ is sufficiently large depending only on the fixed asymptotically AdS spacetime.
\end{proposition}
\begin{remark} \label{unelliptic}
Observe that membership in the space $H^{2,s}_{AdS}\left(\Sigma_\tau\right)$ already imposes strong $r$-decay at infinity. The solution is \underline{non-unique} if this decay is weakened. For instance, one checks that for $\alpha=2$ and $g$ being the pure-AdS metric (with $l=1$, say) $\psi = \frac{c_1}{r} + \frac{c_2 arccot (r)}{r}$ defines a two parameter family of solutions to $L\psi=0$ with both branches decaying to zero. Hence there is non-uniqueness of $(ELP)$ in the class of solutions vanishing at infinity (while only the $c_2$-branch lives in $H^{2,0}_{AdS}\left(\Sigma_\tau\right)$).
\end{remark}
\begin{proof}
The solution may be obtained as a sequence of solutions to the finite problem
\begin{equation} \label{ellPi}
(ELP_i) = \left\{
\begin{array}{rl} L \psi_i = \mathcal{F} & \textrm{in } \Sigma^i_{\tau} \\
\psi_i=0 &  \textrm{for} \ \  r =2^iR \\
\psi_i=0 & \textrm{at} \ \ r=R
\end{array} \right.
\end{equation}
The existence of a unique $H^2$-solution to $(ELP_i)$ is obtained by standard methods as $\Sigma_\tau^i$ is compact, the weights in $r$ are bounded and the operator $L$ has trivial kernel, cf.~footnote \ref{Fredholm}). We will show that any $H^2$-solution to $(ELP)_i$ satisfies the \emph{weighted} estimate
\begin{align} \label{gos}
\| \psi_i \|^2_{H^{2,s}_{AdS} \left(\Sigma_{\tau}^i\right)} \leq C_s \int_{\Sigma_{\tau}^i} \mathcal{F}^2 r^{2+s} dr d\omega 
\end{align}
for a uniform $C_s$ not depending on $i$. To show (\ref{gos}), we first establish\footnote{Note that this immediately implies that the kernel of $L$ is trivial and therefore, by the Fredholm alternative, the existence of a unique weak solution to problem $(ELP_i)$. \label{Fredholm}} 
\begin{align} \label{gos2}
\| \psi_i \|^2_{H^{1,s}_{AdS} \left(\Sigma_{\tau}^i\right)} \leq C_s \int_{\Sigma_{\tau}^i} \mathcal{F}^2 r^{2+s} dr d\omega \, .
\end{align}
Integrating by parts yields 
\begin{align} \label{zeroa}
\frac{C}{\epsilon} \int_{\Sigma_{\tau}^i} \mathcal{F}^2 r^{2+s} dr d\omega + \epsilon \int_{\Sigma_{\tau}^i} \psi_i^2 r^{2+s} dr d\omega \geq -\int_{\Sigma_\tau^i} \mathcal{F} \sqrt{g} r^s \psi_i dr d\omega = \nonumber \\
-\int_{\Sigma_\tau^i} L \psi_i \sqrt{g} r^s \psi_i dr d\omega 
= \int_{\Sigma_\tau^i} \left[g^{rr} \left(\partial_r \psi_i \right)^2 + g^{AB} \partial_A \psi_i \partial_B \psi_i  \right] r^s \sqrt{g} dr d\omega \nonumber \\
- \int_{\Sigma_\tau^i} \left[\frac{\alpha}{l^2} \sqrt{g} r^{s} + \frac{s}{2} \  \partial_r \left( r^{s-1} \sqrt{g} g^{rr} \right)\right] \psi_i^2  dr d\omega \nonumber \\
+ \int_{\Sigma_\tau^i} dr d\omega r^s \ \partial_t \left(g^{tj} \sqrt{g}\right) \left(\partial_j \psi_i\right) \psi_i 
+  \int_{\Sigma_\tau^i} g^{rA} \partial_A \psi_i \partial_r \psi_i r^s \sqrt{g} dr d\omega
\end{align}
for any $\epsilon>0$. Note that the boundary terms vanish in view of $\psi_i=0$ at both boundaries.  For the last line in (\ref{zeroa}) we observe that $|\partial_t \left(g^{tr} \sqrt{g}\right)| \leq \frac{C}{r}$ and $|\partial_t \left(g^{tA} \sqrt{g}\right)| \leq \frac{C}{r^3}$ as well as $|g^{rA}| \leq \frac{C}{r^4}$ for a constant depending only on the fixed background. 

The inequality (\ref{gos2}) would follow if we can absorb the last two lines of (\ref{zeroa}) (and the $\epsilon$-term on the left) by the positive derivative terms in the second line. For $s=0$, this is immediate by Lemma \ref{fila}: One uses the ``good" derivative term on the left of (\ref{filaeq}) to absorb (exploiting the strong $r$-decay as a smallness factor) the $\psi \psi_r$-terms  arising from the last line in (\ref{zeroa}), while the angular derivatives in the second line (and again the good term in Lemma \ref{fila}) absorb the terms of the form $\psi \psi_A$. Finally, the terms of the form $\partial_A \psi \partial_r \psi$ are absorbed in a similar fashion. Note that $\epsilon \rightarrow 0$ for $s \rightarrow \sqrt{9-4\alpha}$.

With (\ref{gos2}) established for $s=0$, it suffices to absorb zeroth order terms in (\ref{zeroa}) which have stronger $r$-weights than $\int r^2 \psi_i^2 dr d\omega$, since the latter are already controlled by the estimate for $s=0$. In particular, for $s\leq 2$, it suffices to show the Hardy inequality
\begin{align} \label{h2}
\frac{1}{l^2} \int_{\Sigma_\tau^i} \left[ \alpha + \frac{1}{2} s\left(s+3\right)  \right] r^{s} \sqrt{g} \psi_i^2 dr d\omega <  \int_{\Sigma_\tau^i} g^{rr} \left(\partial_r \psi_i \right)^2  r^{2+s} \sqrt{g} dr d\omega \, .
\end{align}
Repeating once more the proof of Lemma \ref{fila} one finds that this holds for $s<\sqrt{9-4\alpha}$. This establishes (\ref{gos2}) for $s < \min \left(2,\sqrt{9-4\alpha}\right)$. But then, if $\sqrt{9-4\alpha} \geq 2$ it suffices to absorb zeroth the zeroth order terms in (\ref{zeroa}), which have stronger weights than $\int r^{4-\delta} \psi_i^2 dr d\omega$. Hence it again suffices to establish (\ref{h2}), which we have seen holds for $s<\sqrt{9-4\alpha}$. This proves (\ref{gos2}) for the full range $0 \leq s < \sqrt{9-4\alpha}$.

To prove (\ref{gos}) given the weighted (\ref{gos2}) is then fairly standard. We sketch the argument of Evans \cite{Evans} (pp 317). Consider the tangential part of $L$,
\begin{align}
v_ i = \frac{1}{\sqrt{g}} \partial_A \left( \sqrt{g} g^{AB} \partial_B \psi_i \right) \, .
\end{align}
One can construct a finite differences  approximation of $v_i$, which is seen to be in $H^1_0\left(\Sigma_\tau^i\right)$ because $\psi$ vanishes on the boundary and only ``derivatives" (finite differences) tangential to the the boundary are involved. To emphasize the meat of the argument, we will work with $v_i$ itself and pretend we have enough regularity to justify the following integration by parts (see \cite{Evans} for the computation in terms of the finite differences approximation). We integrate the first term in the identity
\begin{align}
\int_{\Sigma_\tau^i} \sqrt{g} r^s  \left[ g^{kl} \partial_k \psi_i \partial_lv_i  +  v_i \left( \textrm{terms involving $\psi_i$ or $\partial \psi_i$} \right) \right] dr d\omega = \int_{\Sigma_\tau^i} \sqrt{g} r^s  \mathcal{F}v_i dr d\omega  \nonumber 
\end{align}
by parts, moving the angular derivative in $v_i$ onto $\psi_i$. This establishes
\begin{align} \label{ap}
\int_{\Sigma_\tau^i} \sqrt{g} r^s g^{kl} g^{AB} \partial_A \partial_k \psi_i \partial_B \partial_l \psi_i \leq C_s \int_{\Sigma_{\tau}^i} \mathcal{F}^2 r^{2+s} dr d\omega \, ,
\end{align}
where we used that $\psi_i$ and first derivatives of $\psi_i$ are already controlled in the weighted norm (\ref{gos2}) and that one can always borrow a little bit of the $v_i$-norm from the main term on the left of (\ref{ap}). Finally, one estimates the missing $rr$-derivative from writing
\begin{align}
g^{rr} \partial_r \partial_r \psi_i = L \psi + \textrm{terms already estimated by (\ref{gos2}) and (\ref{ap})}
\end{align}
to prove the full (\ref{gos}).

We now show that our sequence of solutions $\psi_i$ to the problem $(ELP_i)$ convergences \emph{strongly} in $H^{2,s}_{AdS}$ to a solution of the original problem $(ELP)$.

For this we consider the difference of two solutions $\psi_{i,k} = \psi_i - \psi_{i+k}$ for any $k\geq 0$. The latter satisfies $L \psi_{i,k}=0$ in the region $\Sigma_\tau^i$, $\psi_{i,k}=0$ on $r=R$, while at $r=2^{i} R$ the quantity $\psi_{i,k}$ is equal to the trace of the solution $\psi_{i+k}$ on this boundary. We introduce the following cut-off version of $\psi_{i+k}$: For any $k\geq 0$,
$\psi_{i+k}^\chi = \left(1-\chi_{i-1}\left(r\right) \right)\cdot \psi_{i+k}$ where $\chi_i$ was defined in Lemma \ref{dens}. Note that $\psi_{i+k}^\chi$ agrees with $\psi_{i+k}$ for $r\geq 2^i R$, while it vanishes for $r \leq 2^{i-1}R$. Clearly,
\begin{align}
\| \psi^\chi_{i+k} \|^2_{H^{2,s}_{AdS} \left(\Sigma^{i+k}_\tau\right)} \leq C \| \psi_{i+k} \|^2_{H^{2,s}_{AdS}\left(\Sigma_\tau^{i+k} \cap \{ r\geq 2^{i-1}R \} \right)} \, .
\end{align}
It follows that we have $L \left(\psi_{i,k} - \psi_{i+k}^\chi\right) = - L \psi_{i+k}^\chi$ in $\Sigma_\tau^i$, with the quantity in the brackets satisfying trivial boundary conditions on both boundaries. From this equation one derives (as before) the estimate
\begin{align}
\|\psi_{i,k}\|^2_{H^{2,s}_{AdS} \left(\Sigma_{\tau}^i\right)} \leq C \|\psi_{i+k}\|^2_{H^{2,s}_{AdS} \left(\Sigma^{i+k}_{\tau} \cap \{ r\geq 2^{i-1}R \} \right)} \, 
\end{align}
for any $k >0$. We estimate the last term by setting $\tilde{s} = s + \frac{1}{2}\left(\sqrt{9-4\alpha} - s\right) < \sqrt{9-4\alpha}$ and estimating
\begin{align}
\|\psi_{i+k}\|^2_{H^{2,s}_{AdS} \left(\Sigma^{i+k}_{\tau} \cap \{ r\geq 2^{i-1}R \} \right)} \leq \left(2^{i-1}R\right)^{s-\tilde{s}} \|\psi_{i+k}\|^2_{H^{2,\tilde{s}}_{AdS} \left(\Sigma^{i+k}_{\tau} \right)}  \nonumber \\ \leq C_{\tilde{s}}\left(2^{i-1}R\right)^{s-\tilde{s}}  \|\mathcal{F}\|^2_{H^{0,\tilde{s}}_{AdS}\left(\Sigma_\tau\right)} \nonumber
\end{align}
for any $k\geq 0$, where the uniform estimate (\ref{gos}) has been used. Clearly, the right hand side goes to zero as $i \rightarrow \infty$. This establishes the existence of an $H^{2,s}$-solution to problem (ELP).

Finally, we show that any solution to $(ELP)$ living in $H_{AdS}^{2,s} \left(\Sigma_\tau\right)$ is unique. For this we have to show that (ELP) for $\mathcal{F}=0$ admits only the zero solution. Consider the difference of two solutions $\tilde{\psi}=\psi_1-\psi_2$, which satisfies $L\tilde{\psi}=0$ and vanishes at infinity and at the boundary. Multiply this by $\sqrt{g}\tilde{\psi}$ and integrate by parts over $\Sigma_{\tau}^i$. This yields (repeating the computations above)
\begin{align}
\| \tilde{\psi} \|_{H^{1,0}_{AdS}\left(\Sigma_{\tau}^i\right)} \leq \| \psi_1\|_{H^{2,0}_{AdS} \left(\Sigma_{\tau} \cap \{r\geq 2^{i-1}R\} \right)} + \| \psi_2\|_{H^{2,0}_{AdS} \left(\Sigma_{\tau} \cap \{r\geq 2^{i-1}R\} \right)} \, ,
\end{align}
where the terms on the right appear from estimating the boundary term at $r=2^iR$ arising in the integration by parts. Since $\psi_1, \psi_2$ are in $H^{2,0}_{AdS}\left(\Sigma_\tau\right)$, the right hand side goes to zero as $i \rightarrow \infty$ establishing that $\tilde{\psi}=0$. 
\end{proof}

For the finite problem with \emph{non-trivial} boundary conditions we can show
\begin{proposition} \label{trela}
Let $\mathcal{F} \in H^{0,s} \left(\Sigma_\tau\right)$ and $u \in H^{2,s}\left(\Sigma_{\tau}\right)$ for any $0 \leq s<\sqrt{9-4\alpha}$ be given. Then on any spacelike slice $\Sigma^i_{\tau} \subset \mathcal{D}$, the elliptic problem 
\begin{equation} \label{ellPf}
(ELP_u) = \left\{
\begin{array}{rl} L \psi = \mathcal{F} & \textrm{in } \Sigma^i_{\tau} \\
\psi=u|_{r=2^iR} &  \textrm{on} \ \ \ r = 2^i R   \\
\psi=0 & \textrm{on} \ \ \ r=R
\end{array} \right.
\end{equation}
has a unique solution in $H^{2,s} \left(\Sigma^i_{\tau}\right)$ satisfying the estimate
\begin{align}
\| \psi \|^2_{H^{2,s}_{AdS} \left(\Sigma^i_{\tau}\right)} \leq C_s \int_{\Sigma_{\tau}} \mathcal{F}^2 r^{2+s} dr d\omega + C \| u \|^2_{H^{2,s}_{AdS} \left(\Sigma_{\tau} \cap \{ r \geq 2^{i-1}R\} \right)}  \, .
\end{align}
provided that $R$ (defining $\mathcal{D}$ as in (\ref{reR})) is chosen sufficiently large depending only on the fixed asymtotically AdS spacetime.
\end{proposition}
Note that the trace of $u$ on the boundary is well-defined in (\ref{ellPf}). Note also that in the limit as $i \rightarrow \infty$ one obtains the result of the previous proposition. 
\begin{proof}
Define the function $\tilde{u} = \left(1-\chi_{i-1} \left(r\right) \right) u$, where $\chi_i$ is the cut-off function defined in the proof of Lemma \ref{dens}. Note that $\tilde{u}$ vanishes for $r \leq 2^{i-1}R$ and agrees with $u$ for $r \geq 2^iR$. Moreover, $\|\tilde{u}\|_{H^{2,s}_{AdS} \left(\Sigma_{\tau}\right)} \leq C \|u\|_{H^{2,s}_{AdS} \left(\Sigma_{\tau} \cap \{r \geq 2^{i-1}R \} \right)}$. We can hence consider the problem $L \left(\psi - \tilde{u}\right) = \mathcal{F} - L \tilde{u}$ with trivial boundary conditions, which has already been dealt with in Proposition \ref{sela}.
\end{proof}
We will, from now on, regard the constant $R$ (and hence the domain $\mathcal{D}$, cf.~(\ref{reR})) as fixed, in particular so that all results of section \ref{elliptic} are true.
\section{The initial data} \label{idata}
\begin{definition}
An \underline{$H^{2}$-initial data set} consist of a triple $\left(u,v,w\right)$ where $v \in H_{AdS}^{1,0} \left(\Sigma_0\right)$, $w \in H_{AdS}^{0,-2} \left(\Sigma_0\right)$ are the free data, and $u$ is the unique $H_{AdS}^{2,0}$-solution vanishing at infinity of the elliptic equation (cf.~Proposition \ref{sela})
\begin{align} \label{dk}
L u = - \frac{1}{\sqrt{g}} \partial_t \left(g^{tt} \sqrt{g} \right) v -  g^{tt}  \cdot  w - 2g^{tj} \partial_j v - \frac{1}{\sqrt{g}} \partial_j \left(g^{tj} \sqrt{g}\right) v =: \mathcal{F}_0 \, .
\end{align}
\end{definition}
Note that with the assumptions on $v$ and $w$, the right hand side of (\ref{dk}) is in fact in $H_{AdS}^{0,2} \left(\Sigma_0\right)$. By Proposition \ref{sela} we obtain that $u \in H_{AdS}^{2,s} \left(\Sigma_0\right)$ for any $0\leq s < \min \left(\sqrt{9-4\alpha},2\right)$ with the estimate
\begin{align} \label{dataes}
\|u\|^2_{H^{2,s}_{AdS}\left(\Sigma_0\right)} \leq C_s \|\mathcal{F}_0\|^2_{H^{0,2}_{AdS}\left(\Sigma_0\right)} \leq C_s   \left(\|v\|^2_{H^{1,0}_{AdS}\left(\Sigma_0\right) } + \|w\|^2_{H^{0,-2}_{AdS}\left(\Sigma_0\right)}\right)  \, .
\end{align}
We now define a sequence approximating the data. Namely, since $v \in H_{AdS}^{1,0} \left(\Sigma_0\right)$ and $w \in H_{AdS}^{0,-2} \left(\Sigma_0\right)$, we can approximate $v$ and $w$ by a sequence of smooth functions of compact support, $v_n$, $w_n$ in their respective spaces. The associated (by~(\ref{dk})) solution $u_n$ then converges to $u$ in $H_{AdS}^{2,s}\left(\Sigma_0\right)$. Note that $u_n$ does \emph{not} have compact support. We summarize this as
\begin{align}
\left(u_n,v_n,w_n\right)\rightarrow \left(u,v,w\right)  \textrm{\ \ \  in $H^{2,s}_{AdS} \left(\Sigma_0\right) \times H^{1,0}_{AdS}  \left(\Sigma_0\right) \times H^{0,-2}_{AdS} \left(\Sigma_0\right)$}
\end{align}
In particular, the estimate (\ref{dataes}) also holds for the approximating sequence.

\begin{remark}
The underlying reason for constructing the data in this fashion instead of specifying $u$ and $v$ as the free data is that even if $u \in H_{AdS}^{2,s}\left(\Sigma_0\right)$ and $v \in H_{AdS}^{1,0}\left(\Sigma_0\right)$, the expression $Lu$ (and hence $w$) is not necessarily in $H_{AdS}^{0,-2}\left(\Sigma_0\right)$ because of the asymptotically AdS $r$-weights in the metric. The improved decay for $Lu$ therefore has to be imposed initially, which is achieved by the above construction. Requiring $w \in H_{AdS}^{0,-2}$ is in turn necessary, since this is what will be propagated by the (commuted) energy estimate.
\end{remark}
\section{The theorem} \label{theosec}
Let $\left(\mathcal{D},g\right)$ be an asymptotically AdS spacetime patch as defined in section \ref{aAdS} and recall that $R$ has been chosen such that the estimates of section 4 hold. Given an $H^{2}$-initial data set $\left(u,v,w\right)$, consider the following initial boundary value problem with Dirichlet boundary conditions:
\begin{equation} \label{probP}
\mathcal{P} = \left\{
\begin{array}{rl} \Box_g \psi - \frac{\alpha}{l^2} \psi = 0 & \text{in } \mathcal{D} \\
\psi=0 & \text{on } \mathcal{I}  \ \ \text{and}  \ {\mathcal{B}_0}\\
\psi|_{\Sigma_0} = u  & \text{ , \ \   } \partial_t \psi|_{\Sigma_0} = v \, .
\end{array} \right.
\end{equation}
We have the following well-posedness statement:
\begin{theorem} \label{mto}
Let $\alpha<\frac{9}{4}$. Given an $H^{2}$-initial data set $\left(u,v,w\right)$, there exists a solution $\psi \left(t,\bold{x}\right)$ to the initial boundary value problem $\mathcal{P}$ in $\mathcal{D}$ with the property that $ \psi \in C\left(\left[0,T\right], H_{AdS}^{2,s} \left(\Sigma_\tau\right) \right)$, $\partial_t \psi \in C\left(\left[0,T\right], H_{AdS}^{1,0} \left(\Sigma_\tau\right) \right)$,  $\partial_t \partial_t \psi \in C\left(\left[0,T\right], H_{AdS}^{0,-2} \left(\Sigma_\tau\right) \right)$ for any $0\leq s < \min\left(2,\sqrt{9-4\alpha}\right)$ and satisfying the estimate
\begin{align}  \label{mainestt}
\| \psi \|^2_{H^{2,s}_{AdS} \left(\Sigma_t\right)} + \| \partial_ t \psi \|^2_{H^{1,0}_{AdS} \left(\Sigma_t\right)} + \| \partial_t \partial_t \psi \|^2_{H^{0,-2}_{AdS} \left(\Sigma_t\right)} \nonumber \\ \leq  \left(\|v\|^2_{H^{1,0}_{AdS}\left(\Sigma_0\right) } + \|w\|^2_{H^{0,-2}_{AdS}\left(\Sigma_0\right)}\right) e^{C_s t} \, . 
\end{align}
\end{theorem}

\begin{theorem} \label{mtounique}
Under the assumptions of Theorem \ref{mto}, the solution $\psi$ is unique within the class of functions in $C\left(\left[0,T\right], H_{AdS}^{1,0} \left(\Sigma_\tau\right) \right)$.
\end{theorem}
We will prove Theorem \ref{mto} by proving the same statement for any member of the sequence of approximating data, $\left(u_n,v_n,w_n\right)\rightarrow \left(u,v,w\right)$, and then argue by density.
\begin{theorem} \label{mto2}
The statement of Theorem \ref{mto} is true for any member of the approximating sequence $\left(u_n,v_n,w_n\right) \rightarrow \left(u,v,w\right)$. In particular, the constant $C_s$ in (\ref{mainestt}) does not depend on $n$.
\end{theorem}

Clearly, Theorem \ref{mto2} implies Theorem \ref{mto}. The former will in turn be proven as follows. Fix $n$. We extend the function $u_n$ from $\Sigma_0$ to $\mathcal{D}$ by imposing that it is constant along the integral curves of $\partial_t$. We will call this function also $u_n$. Consider then the sequence of boundary initial value problems
\begin{equation} \label{probPi}
\mathcal{P}_i = \left\{
\begin{array}{rl} \Box_g \psi_i - \frac{\alpha}{l^2} \psi_i = 0 & \text{in } \mathcal{D}_i \\
\psi_i=u_n|_{\mathcal{B}_i} & \text{on } \mathcal{B}_i  \ \ \text{and}  \ \ \psi_i = 0 \ \ \text{on} \ \ {\mathcal{B}_0}\\
\psi_i|_{\Sigma_0} = u_n  & \text{ , \ \   } \partial_t \psi_i |_{\Sigma_0} = v_n 
\end{array} \right.
\end{equation}
with $i$ so large that $v_n, w_n$ are not supported for $r \geq 2^iR$. For $\mathcal{P}_i$ we prove
\begin{proposition} \label{mnthm}
The (smooth, unique) solution to each $\mathcal{P}_i$ satisfies the estimates
\begin{align}  \label{bios}
\| \psi_i \|^2_{H^{2,s}_{AdS} \left(\Sigma_{t}^i\right)} + \| \partial_ t \psi_i \|^2_{H^{1,0}_{AdS} \left(\Sigma_{t}^i\right)} + \| \partial_t \partial_t \psi_i \|^2_{H^{0,-2}_{AdS} \left(\Sigma_{t}^i\right)} \nonumber \\ 
\leq  e^{C_s t} \left(\|v_n\|^2_{H^{1,0}_{AdS} \left(\Sigma_0\right)} + \|w_n\|^2_{H^{0,-2}_{AdS} \left(\Sigma_0\right)}\right)   \, .  
\end{align}
and\footnote{Note that the right hand side of (\ref{hios}) may go to infinity as $n\rightarrow \infty$. However, for fixed $n$ it is bounded since $v_n$ and $w_n$ are of compact support.}
\begin{align}  \label{hios}
 \| \partial_ t \psi_i \|^2_{H^{2,s}_{AdS} \left(\Sigma_{t}^i\right)} + \| \partial_t \partial_t \psi_i \|^2_{H^{1,0}_{AdS} \left(\Sigma_{t}^i\right)} + \| \partial_t \partial_t \partial_t \psi_i \|^2_{H^{0,-2}_{AdS} \left(\Sigma_{t}^i\right)} \nonumber \\
  \leq C e^{C_s t} \left[\|w_n\|^2_{H^{1,0}_{AdS} \left(\Sigma_{0}\right)} + \|r^2 L v_n\|^2_{H^{0,-2}_{AdS} \left(\Sigma_{0}\right)} +  \|v_n\|_{H^{1,0}_{AdS} \left(\Sigma_{0}\right)} \right] 
\end{align}
in $\mathcal{D}_i$. For the difference of two solutions, $\tilde{\psi}_{i,k} = \psi_{i} - \psi_{i+k}$, we have for any $k>0$ and any $0 \leq s<\tilde{s}<\sqrt{9-4\alpha}$ the following estimate in $\mathcal{D}_i$:
\begin{align}  \label{dife}
\| \tilde{\psi}_{i,k}  \|^2_{H^{2,s}_{AdS} \left(\Sigma_{t}^i\right)}  + \| \partial_t \tilde{\psi}_{i,k}  \|^2_{H^{1,0}_{AdS} \left(\Sigma_{t}^i\right)} + \| \partial_t \partial_t \tilde{\psi}_{i,k}  \|^2_{H^{0,-2}_{AdS} \left(\Sigma_{t}^i\right)} \nonumber \\ 
\leq C \left(T+1\right) e^{C_{\tilde{s}} T} \left(2^{i-1}R\right)^{\frac{s-{\tilde{s}}}{2}} \Big[\|w_n\|^2_{H^{1,0}_{AdS} \left(\Sigma_{0}\right)} + \|r^2 L v_n\|^2_{H^{0,-2}_{AdS} \left(\Sigma_{0}\right)} \nonumber \\ +  \|v_n\|_{H^{1,0}_{AdS} \left(\Sigma_{0}\right)} \Big]   \, .
\end{align}
\end{proposition}
Assuming the validity of the Proposition, we can prove Theorem \ref{mto2}. 
Indeed, for fixed $n$ the square bracket in (\ref{dife}) is bounded, since $v_n$, $w_n$ are of compact support. Hence, together with the factor multiplying the bracket, the right hand side goes to zero as $i \rightarrow \infty$. It follows that the sequence of solutions $\psi_i$ is Cauchy and converges to a solution of problem $\mathcal{P}$ for the initial data $\left(u_n,v_n,w_n\right)$, thereby establishing Theorem \ref{mto2}.

It remains to prove Proposition \ref{mnthm} and the uniqueness assertion of Theorem \ref{mtounique}. This is carried out in the following section.
\section{Proof of Proposition \ref{mnthm} and Theorem \ref{mtounique} } \label{theproof}
Since the region $\mathcal{D}_i$ is compact, standard theory produces a unique smooth solution. The only difficult part is to derive the three estimates, in particular, the correct $r$-weights appearing in them. We introduce the 
energy momentum tensor of the field $\psi$,
\begin{equation} \label{emt}
\mathbb{T}_{\mu \nu} \left[\psi\right] = \partial_\mu \psi \partial_\nu \psi - \frac{1}{2}g_{\mu \nu} \left(g^{\beta \gamma}\partial_\beta \psi \partial_\gamma \psi - \frac{\alpha}{l^2} \psi^2 \right) \, .
\end{equation}
Since the divergence of $\mathbb{T}_{\mu \nu}\left[\psi\right]$ vanishes for solutions to the wave equation, it implies for any spacetime vectorfield $X$ the energy identity
\begin{align} \label{emtd}
\nabla^\mu \left(\mathbb{T}_{\mu \nu} \left[\psi\right] X^\nu\right) = \mathbb{T}_{\mu \nu} {}^{(X)}\pi^{\mu \nu} \, .
\end{align}
Integrated over the spacetime region $\mathcal{D}_i$ this takes the form
\begin{align} \label{emtid}
\int_{\Sigma^i_\tau} \mathbb{T} \left[\psi\right] \left( X, n_{\Sigma}\right) = \int_{\Sigma^i_0} \mathbb{T}  \left[\psi\right]\left( X, n_{\Sigma_0}\right) + \int_{\mathcal{D}_i} \phantom{}^{(X)}\pi \cdot \mathbb{T}  \left[\psi\right] \nonumber \\
+ \int_{\mathcal{B}_i} \mathbb{T} \left[\psi\right]\left(X,n_{\mathcal{B}_i}\right) - \int_{\mathcal{B}_0} \mathbb{T} \left[\psi\right] \left(X,n_{\mathcal{B}_0}\right) \, .
\end{align}
\begin{lemma}
The solution to $\mathcal{P}_i$ satisfies the estimate
\begin{align} \label{bases}
\|\psi_i\|^2_{H_{AdS}^{1,0}\left(\Sigma_{t}^i \right)} + \int_{\Sigma_{t}^i} \left(\partial_t \psi_i\right)^2 dr d\omega \nonumber \\ \leq \frac{C}{R} \int_0^t d\tau \left[ \|\psi_i\|^2_{H_{AdS}^{1,0}\left(\Sigma_{\tau}^i \right)} + \int_{\Sigma_{\tau}^i} \left(\partial_t \psi_i\right)^2 dr d\omega \right] \nonumber \\ + C\left[\|u_n\|^2_{H_{AdS}^{1,0}\left(\Sigma_{0}^i \right)} + \int_{\Sigma_{0}^i} |v_n|^2 dr d\omega\right] + \| u_n \|^2_{H^{1,0}_{AdS} \left(\Sigma_0 \cap \{ r \geq 2^{i-1}R \} \right)}
\end{align}
for a uniform $C$. The term in the second line can be dropped if $\partial_t$ is Killing.
\end{lemma}
\begin{proof}
Apply identity (\ref{emtid}) with $X=\partial_t$. The boundary term on both $\mathcal{B}_0$ and $\mathcal{B}_i$ vanishes in view of the Dirichlet boundary conditions. For the wrong-signed zeroth order term on $\Sigma_\tau^i$ we invoke the Hardy inequality (cf.~Lemma \ref{fila})
\begin{align} \label{hardy}
\int_{R}^{2^i R}\int_{S^2} dr \, d\omega \, \psi_i^2 \sqrt{g|_{\Sigma_{t,i} }} \, g\left(-\partial_t, n_{\Sigma_{t,i}}\right) = \int_{R}^{2^i R}\int_{S^2} dr \, d\omega \, \psi_i^2 \sqrt{g}   \nonumber \\  \leq \frac{4}{9} \int_{R}^{2^i R} \int_{S^2} dr \, d\omega \, \left(\partial_r\psi_i\right)^2\sqrt{g} g^{rr} + \| u_n \|^2_{H_{AdS}^{1,0}  \left(\Sigma_0 \cap \{ r \geq 2^{i-1}R \} \right)} \, ,
\end{align}
where the last term enters from estimating the boundary term $\int_{S^{2}_{t,r_i}} r_i^3 \psi_i^2 d\omega$ which arises in the integration by parts: We use that $\psi_i$ is equal to the trace of the function $u_n$ on $\mathcal{B}_i$ by the boundary condition imposed.
\end{proof}

Next we commute with the vectorfield $\partial_t$. We note (cf.~the appendix of \cite{Mihalisnotes})
\begin{lemma} \label{coml}
Let $\psi$ be a solution of the equation $\Box_g \psi = f$ and $X$ be a vectorfield. Then
\begin{align}
\Box_g \left(X\psi\right) &= X\left(f\right) + \mathcal{C} \left[X\psi\right] \nonumber \\
\mathcal{C} \left[X\psi\right] &= -2 {}^{(X)}\pi^{\alpha \beta} \nabla_\alpha \nabla_\beta \psi - 2 \left[2 \nabla^\alpha {}^{(X)}\pi_{\alpha \mu} - \nabla_\mu \left( tr {}^{(X)}\pi \right) \right] \nabla^\mu \psi \, . \nonumber
\end{align}
\end{lemma}
Applying the Lemma with $X=\partial_t$ and using the decay assumptions on the deformation tensor and its derivatives, we obtain the analogue of (\ref{bases}): 
\begin{align} \label{imes}
 \left[\|\partial_t \psi_i\|^2_{H_{AdS}^{1,0}\left(\Sigma_{t}^i \right)} + \int_{\Sigma_{t}^i} \left(\partial_t \partial_t \psi_i\right)^2 dr d\omega \right] \nonumber \\ \leq \frac{C}{R}  \int_0^t d\tau \left[ \|\partial_t \psi_i\|^2_{H_{AdS}^{1,0}\left(\Sigma_{\tau}^i \right)} + \int_{\Sigma_{\tau}^i} \left(\partial_t \partial_t  \psi_i \right)^2 dr d\omega + \boxed{\|\psi_i\|^2_{H_{AdS}^{2,0}\left(\Sigma_{\tau}^i \right)}}\right] \nonumber \\ +  C \left[\|v_n\|^2_{H_{AdS}^{1,0}\left(\Sigma_{0}^i \right)} +   \|w_n\|^2_{H_{AdS}^{0,-2}\left(\Sigma_{0}^i \right)} \right]\, .
\end{align}
The boxed term arises from the commutation error-term of Lemma \ref{coml}. The entire second line vanishes if $\partial_t$ is Killing.

Observe that this estimate \emph{does not} involve the last term in (\ref{bases}). This is because $\partial_t\psi_i$ vanishes in the trace sense on both boundaries $\mathcal{B}_0, \mathcal{B}_i$ and hence the boundary terms in the Hardy inequality for $\partial_t\psi_i$ all vanish.

Note also that the $L^2$-norm of $\partial_t\psi$ on the left hand side now has a stronger $r$-weight than in (\ref{bases}): The latter only controls $\int \left(\partial_t \psi\right)^2 dr d\omega$ while the right hand side of (\ref{imes}) controls $\int \left(\partial_t \psi\right)^2 r^2 dr d\omega$.

The missing derivatives on the left of (\ref{imes}) are obtained via elliptic estimates for the solution. We write the wave equation $\Box_g \psi + \frac{\alpha}{l^2} \psi = 0$ as
\begin{align}
L \psi_i = \mathcal{F}_t \left[\psi_i\right] \, ,
\end{align}
where the operator $L$ was defined in (\ref{Ldef}) and
\begin{align}
\mathcal{F}_t \left[\psi_i\right] = - \frac{1}{\sqrt{g}} \partial_t \left(g^{tt} \sqrt{g} \right) \partial_t \psi_i -  g^{tt}  \cdot  \partial_t \partial_t \psi_i - 2g^{tj} \partial_j \partial_t \psi_i - \frac{1}{\sqrt{g}} \partial_j \left(g^{tj} \sqrt{g}\right) \partial_t \psi_i \, .\nonumber 
\end{align}
Using the asymptotic behaviour of the metric coefficients, we obtain that 
\begin{eqnarray} \label{rhsest}
\int_{\Sigma_{t}^i} |\mathcal{F}_t\left[\psi_i\right]|^2 r^2 r^s dr d\omega \leq C\left[ \|\partial_t \psi_i\|^2_{H_{AdS}^{1,0}\left(\Sigma_{t}^i \right)} + \int_{\Sigma_{t}^i} \left(\partial_t \partial_t \psi_i\right)^2 dr d\omega \right]
\end{eqnarray}
holds for any $s \leq 2$. Hence from Proposition \ref{trela}  we conclude (taking into account that $\psi_i$ is equal to the trace of the function $u_n$ on the boundary)
\begin{align} \label{impr1}
\| \psi_i\|^2_{H_{AdS}^{2,s}\left(\Sigma_{t}^i \right)} \leq C_s \left[ \|\partial_t \psi_i\|^2_{H_{AdS}^{1,0}\left(\Sigma_{t}^i \right)} + \int_{\Sigma_{t}^i} \left(\partial_t \partial_t \psi_i\right)^2 dr d\omega \right] \nonumber \\ 
+ C \| u_n \|_{H^{2,s}_{AdS} \left(\Sigma_0 \cap \{ r \geq 2^{i-1}R \} \right)}
 \end{align}
 for any $s<\min\left(2,\sqrt{9-4\alpha}\right)$, where the constant may blow up as $s$ approaches this value. For the last term we can invoke the estimate (\ref{dataes}). Finally, adding the estimates (\ref{impr1}), (\ref{imes}) and (\ref{bases}) we obtain the first estimate of Proposition \ref{mnthm} after applying Gronwall's inequality.
 The second estimate of Proposition \ref{mnthm} follows by commuting the wave equation once more with $\partial_t$. This will produce (cf.~Lemma \ref{coml})
 \begin{align}
 \Box_g \left(\partial_t \partial_t \psi_i \right) + \frac{\alpha}{l^2} \left(\partial_t \partial_t \psi_i \right) = \partial_t  \mathcal{C} \left[\partial_t \psi_i\right] +  \mathcal{C} \left[\partial_t \partial_t \psi_i\right]  \, ,
 \end{align} 
 with the right hand side vanishing in the case that $\partial_t$ is Killing. In general we have the error-estimate
 \begin{align}
 \int_{\mathcal{D}_i} \left(| \partial_t  \mathcal{C} \left[\partial_t \psi_i\right] |^2 +  | \mathcal{C} \left[\partial_t \partial_t \psi_i\right]  |^2 \right) \sqrt{g} r^s dt dr d\omega \leq \nonumber \\
 \frac{C}{R} \int_0^t d\tau \left[ \|\partial_t \partial_t  \psi_i\|^2_{H_{AdS}^{1,0}\left(\Sigma_{\tau}^i \right)} + \|\partial_t \partial_t \partial_t \psi_i\|^2_{H_{AdS}^{0,-2}\left(\Sigma_{\tau}^i \right)}  + \| \partial_t   \psi_i\|^2_{H_{AdS}^{2,0}\left(\Sigma_{\tau}^i \right)}\right] 
 \end{align}
as a consequence of the asymptotic decay of the metric and the fact that any term containing three derivatives must have at least one $\partial_t$-derivative in it. With this we obtain the analogue of (\ref{imes}),
\begin{align} \label{imes2}
 \left[\|\partial_t \partial_t\psi_i\|^2_{H_{AdS}^{1,0}\left(\Sigma_{t}^i \right)} + \int_{\Sigma_{t}^i} \left(\partial_t \partial_t \partial_t\psi_i\right)^2 dr d\omega \right] \nonumber \\ \leq \frac{C}{R}  \int_0^t d\tau \left[ \|\partial_t \partial_t\psi_i\|^2_{H_{AdS}^{1,0}\left(\Sigma_{\tau}^i \right)} + \|\partial_t \partial_t \partial_t\psi_i\|^2_{H_{AdS}^{0,-2}\left(\Sigma_{\tau}^i \right)} +\|\psi_i\|^2_{H_{AdS}^{2,0}\left(\Sigma_{\tau}^i \right)}\right] + \nonumber \\  C \left[\|w_n\|^2_{H_{AdS}^{1,0}\left(\Sigma_{0}^i \right)} +   \|r^2 L v_n \|^2_{H_{AdS}^{0,-2}\left(\Sigma_{0}^i \right)} + \|u_n\|_{H_{AdS}^{2,0}\left(\Sigma_{0}^i \right)} + \|v_n\|_{H_{AdS}^{1,0}\left(\Sigma_{0}^i \right)}  \right] ,
\end{align}
where all terms on the right hand side except the first two in the third line would vanish if $\partial_t$ is Killing. The term $\|u_n\|_{H_{AdS}^{2,0}\left(\Sigma_{0}^i \right)}$ may be dropped in view of (\ref{dataes}). Turning to the elliptic estimate, this time derived from 
 \begin{align}
 L \left(\partial_t \psi_i\right) = \mathcal{F}_t \left[\partial_t \psi_i\right] - \mathcal{C} \left[\partial_t \psi_i\right] 
 \end{align}
 we produce the analogue of (\ref{impr1}), 
 \begin{align} \label{impr2}
\| \partial_t \psi_i\|^2_{H_{AdS}^{2,s}\left(\Sigma_{t}^i \right)} \leq C_s \left[ \|\partial_t \partial_t \psi_i\|^2_{H_{AdS}^{1,0}\left(\Sigma_{t}^i \right)} +  \| \partial_t \partial_t  \partial_t  \psi_i\|^2_{H_{AdS}^{0,-2}\left(\Sigma_{t}^i \right)}  + \boxed{\| \psi_i\|^2_{H_{AdS}^{2,0}\left(\Sigma_{\tau}^i \right)}} \right]
\end{align}
Note that the last term in (\ref{impr1}) is absent for (\ref{impr2}). This is because $\partial_t \psi_i$ vanishes on $\mathcal{B}_i$ (whereas $\psi_i$ did not). Combining (\ref{imes2}) and (\ref{impr2}) one arrives at (\ref{hios}) after applying Gronwall's inequality.
 
It remains to derive estimates for the difference of two solutions. For this, we will first establish (\ref{difepre}). Another commutation and the same argument repeated will eventually yield (\ref{dife}).

For any $k>0$, the difference of two solutions arising from problem $\mathcal{P}_i$ and $\mathcal{P}_{i+k}$ respectively satisfies
\begin{equation}
\Box_g \tilde{\psi}_{i,k} + \frac{\alpha}{l^2} \tilde{\psi}_{i,k} = 0 \textrm{ \ \ \ with \ \ \ $\tilde{\psi}_{i,k} = \psi_i - \psi_{i+k}$}
\end{equation} 
in $\mathcal{D}_{i}$. The initial data is vanishing, as is the boundary data on $\mathcal{B}_0$. The boundary data on $\mathcal{B}_{i}$, on the other hand, is equal to $u_n |_{\mathcal{B}_i}$ minus the trace of the solution $\mathcal{P}_{i+k}$ induced on that boundary. Hence we will obtain the energy estimate
\begin{eqnarray} \label{aux}
\|  \tilde{\psi}_{i,k} \|^2_{H_{AdS}^{1,0}\left(\Sigma^i_t\right)} + \| \partial_t  \tilde{\psi}_{i,k}  \|^2_{H_{AdS}^{0,-2}\left(\Sigma^i_t\right)} \nonumber \\ 
\leq \frac{C}{R} \int_{\mathcal{D}_i} dt \left[\|  \tilde{\psi}_{i,k} \|^2_{H_{AdS}^{1,0}\left(\Sigma^i_t\right)} + \| \partial_t   \tilde{\psi}_{i,k} \|^2_{H_{AdS}^{0,-2}\left(\Sigma^i_t\right)}\right] + C \Big|\int_{\mathcal{B}_{i}} \mathbb{T} \left[  \tilde{\psi}_{i,k}\right]\left(\partial_t, n_{\mathcal{B}_i}\right) \Big| \, .
\end{eqnarray}
For the boundary term in (\ref{aux}) we will use that every $\psi_i$ is in $H_{AdS}^{2,s}\left(\Sigma_{t}^i\right)$ and $\partial_t \psi_i$ is in $H_{AdS}^{1,0}\left(\Sigma_{t}^i\right)$: Let $\tilde{\chi}_{i-1}\left(r\right)=1-\chi_{i-1}\left(r\right)$ (recall $\chi_i\left(r\right)$ defined in the proof of Lemma \ref{dens}), which is $1$ on $\mathcal{B}_{i}$ and $0$ on $\mathcal{B}_{i-1}$ and satisfies $|\tilde{\chi}_{i-1}^\prime \left(r\right)| \leq \frac{C}{r}$. Applying the energy identity (\ref{emtd}) in the region $\mathcal{D}_{i} \setminus \mathcal{D}_{i-1}$ with the vectorfield $X=\tilde{\chi}_{i-1} \cdot \partial_t$ leads to
\begin{eqnarray} \label{bndtermcomp}
 \int_{\mathcal{B}_{i}} \partial_t \left(\tilde{\psi}_{i,k}\right) n_{\mathcal{B}_i} \left(\tilde{\psi}_{i,k}\right)  = \int_{\Sigma^i_t} \tilde{\chi}_{i-1} \,  \mathbb{T}_{\mu t}\left[\tilde{\psi}_{i,k}\right] n^\mu  - \int_{\Sigma_0^i} \tilde{\chi}_{i-1} \,  \mathbb{T}_{\mu t}\left[\tilde{\psi}_{i,k}\right] n^\mu \nonumber \\  - \int_{\mathcal{D}_{i} \setminus \mathcal{D}_{i-1}} \tilde{\chi}_{i-1}\phantom{}^{(T)}\pi \cdot  \mathbb{T}\left[\tilde{\psi}_{i,k}\right] +  \int_{\mathcal{D}_{i} \setminus \mathcal{D}_{i-1}} \left(\tilde{\chi}_{i-1}\right)_{,r} g^{r\mu} \, \mathbb{T}_{\mu t}\left[\tilde{\psi}_{i,k}\right] .
 \end{eqnarray}
 From  this we derive the estimate
\begin{align} \label{baux}
\Big| \int_{\mathcal{B}_{i}} \partial_t \left(\tilde{\psi}_{i,k}\right) n_{\mathcal{B}_i} \left(\tilde{\psi}_{i,k}\right)\Big|   \leq C\left(2^{i-1}R\right)^{-\frac{s}{2}} \left(T+1\right) \sup_{t \in [0,T]}\Bigg[\| \psi_i \|^2_{H_{AdS}^{2,s}\left(\Sigma^i_t\right)}  \nonumber \\ + \| \partial_t \psi_{i} \|^2_{H_{AdS}^{1,0}\left(\Sigma^i_t\right)} + \| \psi_{i+k} \|^2_{H_{AdS}^{2,s}\left(\Sigma^i_t\right)} + \| \partial_t \psi_{i+k} \|^2_{H_{AdS}^{1,0}\left(\Sigma^{i}_t\right)} \Bigg] \, .
\end{align}
Note that the right hand side goes to zero for $i\rightarrow \infty$, 
since the square bracket is uniformly bounded from initial data by the first estimate of Proposition \ref{mnthm}.
The only difficult term to control in order to obtain the previous estimate is the last term in (\ref{bndtermcomp}), with the contraction taken in $r$. It is estimated
\begin{align} \label{hc}
 \Big|\int_{\mathcal{D}_{i} \setminus \mathcal{D}_{i-1}} \left(\tilde{\chi}_{i-1}\right)_{,r} r^4 \left(\partial_t \tilde{\psi}_{i,k}\right) \left(\partial_r \tilde{\psi}_{i,k} \right) dt dr d\omega \Big| 
 \nonumber \\ 
 \leq C \int_{\mathcal{D}_{i} \setminus \mathcal{D}_{i-1}} r^3 \Big| \partial_t \tilde{\psi}_{i,k} \Big|\Big|(\partial_r \tilde{\psi}_{i,k}\Big| dt dr d\omega \nonumber \\  \leq C \cdot T \cdot r^{-\frac{s}{2}}_{i-1}  \sup_{t \in [0,T]}\Bigg[\| \psi_i \|^2_{H_{AdS}^{2,s}\left(\Sigma^i_t\right)} + \| \partial_t \psi_{i} \|^2_{H_{AdS}^{1,0}\left(\Sigma^i_t\right)} \nonumber \\ + \| \psi_{i+k} \|^2_{H_{AdS}^{2,s}\left(\Sigma^i_t\right)} + \| \partial_t \psi_{i+k} \|^2_{H_{AdS}^{1,0}\left(\Sigma^{i}_t\right)} \Bigg]   \, .
\end{align}

\begin{remark}
Note that, while being of the same order of derivatives, the term in the second line of (\ref{hc}) cannot be controlled by the energy norm we are estimating on the left hand side of (\ref{aux}) as the $r$-weight of that term is too strong. Indeed, the first order energy only controls $\int \left(\partial_t \tilde{\psi}_{i,k}\right)^2 dr \, d\omega$, whereas the $\partial_t$-commuted energy controls in particular $\int \left(\partial_t \tilde{\psi}_{i,k}\right)^2 r^2 dr \, d\omega$ as a zeroth order term, which is what is necessary to control that term. Similarly, the first order energy only controls $\int \left(\partial_r \tilde{\psi}_{i,k}\right)^2 r^4 dr \, d\omega$, while the improved elliptic estimate arising from the $\partial_t$-commuted energy controls $\int \left(\partial_r \tilde{\psi}_{i,k}\right)^2 r^{4+s} dr \, d\omega$. This additional gain in $r$-decay provides the smallness factor of $r^{-\frac{s}{2}}$ in the estimate and ensures convergence of the sequence of solutions $\psi_i$.
\end{remark}
Combining (\ref{aux}) and (\ref{baux}) yields after inserting (\ref{bios}) and applying Gronwall's inequality
\begin{align}  \label{difepre}
\| \tilde{\psi}_{i,k} \|^2_{H^{1,0}_{AdS} \left(\Sigma_{t}^i\right)}  + \| \partial_t \tilde{\psi}_{i,k}  \|^2_{H^{0,-2}_{AdS} \left(\Sigma_{t}^i\right)} \nonumber \\
\leq e^{C_sT} \cdot T \left(2^{i-1}R\right)^{-\frac{s}{2}} \Big[
\|v_n\|^2_{H^{1,0}_{AdS} \left(\Sigma_0\right)} + \|w_n\|^2_{H^{0,-2}_{AdS} \left(\Sigma_0\right)}  \Big] \, .
\end{align}
Finally, for (\ref{dife}) we start from
\begin{align}
\Box_g \left(\partial_t \tilde{\psi}_{i,k} \right) + \frac{\alpha}{l^2}  \left(\partial_t  \tilde{\psi}_{i,k} \right) = \mathcal{C}\left[\partial_t  \tilde{\psi}_{i,k}\right] \nonumber \, 
\end{align} 
to derive the analogue of (\ref{aux}). Because of the inhomogeneity (which vanishes if $\partial_t$ is Killing in $\mathcal{D}$), it now reads
\begin{align} \label{aux2}
\| \partial_t \tilde{\psi}_{i,k} \|^2_{H_{AdS}^{1,0}\left(\Sigma^i_t\right)} + \| \partial_t \partial_t \tilde{\psi}_{i,k}  \|^2_{H_{AdS}^{0,-2}\left(\Sigma^i_t\right)} \nonumber \\ 
\leq \frac{C}{R} \int_{\mathcal{D}_i} dt \Bigg[\| \partial_t \tilde{\psi}_{i,k}  \|^2_{H_{AdS}^{1,0}\left(\Sigma^i_t\right)} + \| \partial_t \partial_t \tilde{\psi}_{i,k}  \|^2_{H_{AdS}^{0,-2}\left(\Sigma^i_t\right)}  + \boxed{\| \tilde{\psi}_{i,k}  \|_{H^{2,0}_{AdS} \left(\Sigma_\tau^i\right)}} \Bigg] \nonumber \\ + C \Big|\int_{\mathcal{B}_{i}} \mathbb{T} \left[\partial_t \tilde{\psi}_{i,k}  \right]\left(\partial_t, n_{\mathcal{B}_i}\right) \Big| \, .
\end{align}
with the boxed term taking into account the inhomogeneity and the entire second line vanishing if $\partial_t$ is Killing. For the boundary term on $\mathcal{B}_i$ we are going to use exactly the same argument as above, now producing the estimate
\begin{align} \label{bupd}
\Big|\int_{\mathcal{B}_{i}} \mathbb{T} \left[\partial_t \tilde{\psi}_{i,k}\right]\left(\partial_t, n_{\mathcal{B}_i}\right) \Big| \leq C \cdot \left(T+1\right) \cdot \left(2^{i-1}R\right)^{-\frac{s}{2}}  \sup_{t \in [0,T]}\Bigg[\| \partial_t\psi_i \|^2_{H_{AdS}^{2,s}\left(\Sigma^i_t\right)} \nonumber \\  + \| \partial_t \partial_t \psi_{i} \|^2_{H_{AdS}^{1,0}\left(\Sigma^i_t\right)} + \| \partial_t \psi_{i+k} \|^2_{H_{AdS}^{2,s}\left(\Sigma^i_t\right)} + \| \partial_t \partial_t \psi_{i+k} \|^2_{H_{AdS}^{1,0}\left(\Sigma^{i}_t\right)} \Bigg] \, .
\end{align}
Next we write the wave equation for $\tilde{\psi}_{i,k}$ in elliptic form on each $\Sigma_{t}^i$. We have
\begin{align}
L \tilde{\psi}_{i,k} = \mathcal{F}_t \left[\tilde{\psi}_{i,k} \right] \textrm{ \ \ \ in $\Sigma_\tau^i$}
\end{align}
and $\tilde{\psi}_i = 0$ on $\mathcal{B}_0$, while $\tilde{\psi}_i$ is equal to $u_n|_{\mathcal{B}_i}$ minus the trace of $\psi_{i+k}$ on $\mathcal{B}_i$. Hence applying Proposition \ref{trela} yields
\begin{align} \label{elf}
\|\tilde{\psi}_{i,k} \|^2_{H^{2,s}_{AdS} \left(\Sigma_\tau^i\right)} \leq C_s \left[ \| \partial_t \tilde{\psi}_{i,k}  \|^2_{H^{1,0}_{AdS} \left(\Sigma_{t}^i\right)} + \| \partial_t \partial_t \tilde{\psi}_{i,k}  \|^2_{H^{0,-2}_{AdS} \left(\Sigma_{t}^i\right)}\right] \nonumber \\
 + C \|u_n\|^2_{H^2_{AdS} \left(\Sigma_0 \cap \{ r \geq 2^{i-1}R \} \right) } + C \| \psi_{i+k}\|^2_{H^{2,s}_{AdS} \left(\Sigma^{i+k}_t \cap \{ r \geq 2^{i-1}R \}\right) }\, . 
\end{align}
We note that the last term can be estimated using the uniform estimate (\ref{bios}),
\begin{align}
 \| \psi_{i+k}\|^2_{H^{2,s}_{AdS} \left(\Sigma^{i+k}_t \cap \{ r \geq 2^{i-1}R \}\right) } \leq  \left(2^{i-1}R\right)^{s-\tilde{s}} \| \psi_{i+k}\|^2_{H^{2,\tilde{s}}_{AdS} \left(\Sigma^{i+k}_t\right) } \nonumber \\ \leq C \left(2^{i-1}R\right)^{s-\tilde{s}} e^{C_{\tilde{s}} t} \left(\|v_n\|^2_{H^{1,0}_{AdS} \left(\Sigma_0\right)} + \|w_n\|^2_{H^{0,-2}_{AdS} \left(\Sigma_0\right)}\right)  \, ,
\end{align}
for any $\tilde{s}$ with $s<\tilde{s}<\sqrt{9-4\alpha}$. A similar estimate holds from (\ref{dataes}) for the penultimate term in (\ref{elf}).
We summarize by combining (\ref{elf}) with (\ref{aux2}), (\ref{bupd}) and inserting the uniform estimate (\ref{hios}) for the square bracket in (\ref{bupd}) to arrive at the inequality
\begin{align} \label{gw}
f\left(t\right) \leq C_s \int_0^t f\left(z\right) dz + \textrm{right hand side of (\ref{dife})} \, 
\end{align}
for
\begin{align}
f \left(t\right) = \|\tilde{\psi}_{i,k} \|^2_{H^{2,s}_{AdS} \left(\Sigma_i\right)} + \| \partial_t \tilde{\psi}_{i,k} \|^2_{H_{AdS}^{1,0}\left(\Sigma^i_t\right)} + \| \partial_t \partial_t \tilde{\psi}_{i,k}  \|^2_{H_{AdS}^{0,-2}\left(\Sigma^i_t\right)} \, .
\end{align}
Applying Gronwall's inequality to (\ref{gw}) yields (\ref{dife}).
\subsection{Proof of Theorem \ref{mtounique}}
Suppose we have two solutions of $\mathcal{P}$, $\psi_1$ and $\psi_2$ living in the spaces of Theorem \ref{mto}. Consider the wave equation for the difference, $\tilde{\psi}=\psi_1 - \psi_2$, and apply the standard energy estimate in $\mathcal{D}_i$. Redoing the computation that lead to (\ref{difepre}) we find
\begin{align}
\| \tilde{\psi} \|_{H^{1,0}_{AdS} \left(\Sigma_\tau^i\right)} \leq \frac{C}{R} \int dt \| \tilde{\psi} \|^2_{H^{1,0}_{AdS} \left(\Sigma_\tau^i\right)} + C \sum_{k=1,2} \| \psi_k \|^2_{H^{2,0}_{AdS} \left(\Sigma_\tau \{ r \geq 2^{i-1}R \} \right)} \, ,
\end{align}
Since both $\psi_1,\psi_2$ are in $H_{AdS}^{2,0}$, the last term goes to zero as $i \rightarrow \infty$ establishing that $\tilde{\psi}=0$ up to some time $t^\star >0$ depending only on the background. One can then reiterate the argument finitely many times to establish uniqueness up to $T$.

\subsection{Uniqueness vs Non-Uniqueness} \label{uniqueness}
As for the elliptic problem, uniqueness only holds in the class of $H^{2,0}_{AdS}$ solutions, cf.~Remark \ref{unelliptic}. We may paraphrase this by saying that there is only one solution in the energy class arising from the approximate Killing field $\partial_t$, as this naturally corresponds to the space $H^{2,0}_{AdS}$. 

This opens the possibility of solutions exhibiting weaker decay at infinity but nevertheless satisfying the Dirichlet condition $\psi=0$ at $\mathcal{I}$. Let us finally show that such alternative solutions do indeed exist and consider the case $\alpha=2$ as an example.\footnote{The following illustrative example was suggested to me by an anonymous referee of \cite{HolzegelAdS}.}

\begin{example} \label{exmp}
We let $\alpha=2$ and consider pure AdS as a background with the additional assumption of spherical symmetry on $\psi$. The coordinate transformation $r = l \frac{\cos x}{\sin x}$ takes null-infinity to $x=0$ and the origin $r=0$ to $\frac{\pi}{2}$. Substituting $\tilde{\psi} = r \psi$ the equation $\Box_g \psi + 2\psi=0$ reduces to the study of the two dimensional wave equation
\begin{equation}
\partial_t^2 \tilde{\psi} = \partial_x^2 \tilde{\psi} \, ,
\end{equation}
with data given on the interval $[\frac{\pi}{2},0]$. Let us for the moment forget about the boundary condition and instead extend the solution to negative values of $x$. Solving the wave equation with arbitrary data on $[-\frac{\pi}{2},\frac{\pi}{2}]$ at $t=0$, say, the solution will in general have a non-vanishing trace on the set $\{t,x=0\}$. If we extend the solution as an odd function, however, then because the symmetry is preserved in evolution, $\psi$ will also vanish on $x=0$ for all times. Extending the data in two different ways we have obtained two different solutions $\tilde{\psi}_1$, $\tilde{\psi}_2$. However, since $\psi = \frac{\tilde{\psi}}{r}$, both of the associated solutions $\psi_1$ and $\psi_2$ satisfy the original Dirichlet boundary initial value problem. In fact we see that one of the solutions decays like $\frac{1}{r^2}$ (corresponding to the one of Theorem \ref{mnthm}), while the other decays only like $\frac{1}{r}$ near infinity (and has hence infinite $\partial_t$-energy flux through $\mathcal{I}$).
\end{example}

\subsection{Improved estimates}
Theorem \ref{mto} provides us with a unique solution in a weighted $H^2$-space. Clearly, using further commutations one can establish regularity in weighted $H^k$-spaces for arbitrary high $k$, provided this regularity is assumed initially. What is perhaps not quite as immediate is that for the range $2\leq s<\sqrt{9-4\alpha}$ higher regularity implies an improvement of the $r$-weight in the $H_{AdS}^{2,s}$-norm for $\psi$. This is seen as follows. Assume $2\leq s<\sqrt{9-4\alpha}$. Applying Theorem \ref{mto} to the initial boundary value problem for $\Box_g \left(\partial_t \psi\right) = \mathcal{C}\left[\partial_t \psi\right]$ yields\footnote{Strictly speaking, Theorem \ref{mto} only applies to the homogenous problem. However, the right hand side is easily dealt with repeating the estimates carried out in section \ref{theproof}.} in particular that $\partial_t \partial_t \psi \in H_{AdS}^{1,0} \left(\Sigma_\tau\right)$. This improves the estimate for the right hand side in the old elliptic estimate for $\psi$ derived from $L\psi = \mathcal{F}_t \left[\psi\right]$. Indeed, now $\mathcal{F}_t \left[\psi\right] \in H_{AdS}^{0,4}\left(\Sigma_\tau\right)$ (whereas before we only had $\partial_t \partial_t \psi \in H_{AdS}^{0,-2} \left(\Sigma_\tau\right)$ which gave $\mathcal{F}_t \left[\psi\right] \in H_{AdS}^{0,2}\left(\Sigma_\tau\right)$, cf.~(\ref{rhsest})). This allows us to apply Proposition \ref{sela} with the full $s<\sqrt{9-4\alpha}$ instead of $s<\min \left(2,\sqrt{9-4\alpha}\right)$ and establish that $\psi \in H_{AdS}^{2,s} \left(\Sigma_\tau\right)$ for the full range $s<\sqrt{9-4\alpha}$. Using Sobolev embedding, we derive the pointwise decay bound $\|\psi\| \leq C_s \cdot r^{-\frac{3}{2}-\frac{1}{2}s}$ for any $s<\sqrt{9-4\alpha}$. In this way one retrieves the class of solutions considered in \cite{HolzegelAdS}, Definition 3.1. 
\section{Applications} \label{finsec}

\subsection{Globalizing the result} \label{gloco}
Theorem \ref{mto} can be combined with the usual local well-posedness result for the wave equation to produce a \emph{spatially global} well-posedness statement for a suitable class of spacetimes.\footnote{In view of the linearity of the wave equation global existence in time is immediate.}

Assume, for instance, that $\left(\mathcal{M},g_{M,l,a}\right)$ is a Kerr-AdS spacetime with 
spacelike slice $\Sigma_0$, on which data $\left(\psi_0, \psi_0^\prime\right)$ for the massive massive equation are defined (with $\psi_0$ and $\psi^\prime_0$ being constructed as in section \ref{idata}). We first solve the wave equation in the domain of dependence of $\Sigma_0$, which is standard. This induces (perhaps non-trivial) boundary data on the hypersurface $r=R\left(M,l,a\right)$, the $R$ being as large as required for the results of this paper to apply. Since Theorem \ref{mto} obviously also holds with non-trivial Dirichlet data imposed at $r=R$, we can apply it to produce a solution in the full $\mathcal{M} \cap J^+ \left(\Sigma_0\right)$. Alternatively, if one wants to apply Theorem \ref{mto} directly  (trivial boundary conditions at $r=R$), one can decompose any given data set into a set supported for $r < 3R$ and another supported for $r>2R$ only, such that their sum equals the original data.  For the latter set one applies Theorem \ref{mto}: By finite speed of propagation the solution is actually zero in a neighborhood of $r=R$ for small times and hence may trivially be extended by zero to the left. For the first set, the solution in the domain of dependence can  -- for small times -- be trivially extended to zero in a neighborhood of null-infinity. Adding the two solutions and using the linearity of the wave equation one produces the desired global solution for the given data in a small time interval.
\begin{corollary}
If $\left(\mathcal{M},g\right)$ is a pure AdS or Kerr-AdS spacetime, then the massive wave equation is globally well-posed in the energy space $H_{AdS}^{2,0}$ provided $\alpha<\frac{9}{4}$.
\end{corollary}
In fact, is is not hard to see that the result holds for any asymptotically AdS spacetime with interior global causal structure similar to one of the cases above.
\subsection{Spherical Symmetry}
The results proven here will be applied in a non-linear spherically-symmetric context in \cite{gs:lwp,gs:stab}. With that in view, let us examine the regularity assumptions on the metric required for Theorem \ref{mto}. Revisiting the proof carefully, we observe that the constant $C_s$ in the estimate of Theorem \ref{mto} (and in (\ref{bios}) respectively) depends only on \emph{second} derivatives of the metric
(while the constants $C,C_s$ in the higher order estimates (\ref{hios}) and (\ref{dife}) require three derivatives). More precisely, the metric should be $C^1$ \emph{and} we need to be able to estimate the error arising from one commutation with the vectorfield $\partial_t$, $\mathcal{C} \left[\partial_t \psi\right]$, which generically involves two derivatives of the metric. In the spherically symmetric application we have in mind, the metric will be written in null-coordinates $u$, $v$ as 
\begin{align} \label{sef}
g = -\Omega^2 \left(u,v\right) du dv + r^2\left(u,v\right) d\Sigma^2
\end{align}
with $r\left(u,v\right)$ being the area radius and appropriate asymptotic conditions on $\Omega$ and $r$ ensuring that the metric is asymptotically AdS. One shows that in this setting the components of the deformation tensor of the vectorfield $\partial_u + \partial_v$ (=$\partial_t$) are
\begin{align}
\pi^{uu} = \pi^{vv} = 0 \textrm{ \ \ \ , \ \ \ } \pi^{uv} = -\frac{2}{\Omega^3} \left(\Omega_u + \Omega_v\right) \textrm{ \ \ \ , \ \ \ } \pi^{AB} = \frac{1}{r} g^{AB} \left(r_u + r_v\right) \nonumber
\end{align}
Evaluating the expression  from Lemma \ref{coml} we obtain
\begin{align}
\mathcal{C} \left[\partial_u \psi + \partial_v \psi\right] = -4 {}^{(X)}\pi^{uv} \partial_u \partial_v \psi - 2 \left[2 \partial_u \pi^{uv} - g^{uv} \partial_u \left(2 g_{uv} \pi^{uv}  + g_{AB} \pi^{AB} \right) \right] \partial^v \psi \nonumber \\
- 2 \left[2 \partial_v \pi^{uv} - g^{uv} \partial_v \left(2 g_{uv} \pi^{uv}  + g_{AB} \pi^{AB} \right) \right] \partial^u \psi \nonumber \\
+ \textrm{terms involving one-derivative of $\psi$ and/or the metric} \nonumber
 \end{align}
In particular, we note that the expression does not involve second derivatives of $\Omega$ (but, of course, second derivatives of $r$). 
It follows that in the spherically symmetric case, for metrics of the form (\ref{sef}), Theorem \ref{mto} holds assuming that $r$ is $C^2$ and $\Omega$ is $C^1$. 

\section{Acknowledgements}
I would like to thank Jacques Smulevici and Claude Warnick for very helpful discussions and a careful reading of the manuscript. I am also grateful to Pedro Gir\~ao, Jos\'e Nat\'ario and Jorge Silva for their interest and detailed comments leading to several improvements of this manuscript. Finally, an anonymous referee is thanked for further insightful comments.
 
\bibliographystyle{utphys}
\bibliography{thesisrefs}

\providecommand{\href}[2]{#2}\begingroup\raggedright\begin{thebibliography}{10}

\bibitem{HolzegelAdS}
G.~Holzegel, ``{On the massive wave equation on slowly rotating Kerr-AdS
  spacetimes},'' {\em Comm.~Math.~Phys.} {\bf 294} (2010) 169--197,
\href{http://arXiv.org/abs/arXiv:0902.0973}{{\tt arXiv:0902.0973}}.

\bibitem{Mihalisnotes}
M.~Dafermos and I.~Rodnianski, ``{Lectures on black holes and linear waves},''
  {\em {Institut Mittag-Leffler Report no. 14, 2008/2009}} (2008)
  \href{http://arXiv.org/abs/arXiv:0811.0354}{{\tt arXiv:0811.0354}}.

\bibitem{Vasy2}
A.~Vasy, ``{The wave equation on asymptotically Anti-de Sitter spaces},'' {\em
  to appear in Analysis and PDE} (2009)
  \href{http://arXiv.org/abs/arXiv:0911.5440}{{\tt arXiv:0911.5440}}.

\bibitem{Breitenlohner}
P.~Breitenlohner and D.~Z. Freedman, ``{Stability in Gauged Extended
  Supergravity},'' {\em Ann. Phys.} {\bf 144} (1982)
249.

\bibitem{Bachelot}
A.~Bachelot, ``{The Dirac System on the Anti-de Sitter Universe},'' {\em
  Commun.~Math.~Phys.} {\bf 283} (2008) 127--167.

\bibitem{gs:lwp}
G.~Holzegel and J.~Smulevici, ``{Self-gravitating Klein-Gordon fields in
  asymptotically Anti-de Sitter spacetimes},'' {\em to appear in Annales Henri
  Poincar\'e} (2011) \href{http://arXiv.org/abs/arXiv:1103.0712}{{\tt
  arXiv:1103.0712}}.

\bibitem{Henneaux}
M.~Henneaux and C.~Teitelboim, ``{Asymptotically anti-De Sitter Spaces},'' {\em
  Commun. Math. Phys.} {\bf 98} (1985)
391--424.

\bibitem{Hollands2}
S.~Hollands, A.~Ishibashi, and D.~Marolf, ``{Comparison between various notions
  of conserved charges in asymptotically AdS-spacetimes},'' {\em Class. Quant.
  Grav.} {\bf 22} (2005) 2881--2920,
\href{http://arXiv.org/abs/hep-th/0503045}{{\tt hep-th/0503045}}.

\bibitem{Evans}
L.~C. Evans, {\em Partial Differential Equations}.
\newblock American Mathematical Society, Providence, RI, 1998.

\bibitem{gs:stab}
G.~Holzegel and J.~Smulevici, ``{Stability of Schwarzschild-AdS for the
  spherically symmetric Einstein-Klein Gordon system},''
  \href{http://arXiv.org/abs/arXiv:1103.3672}{{\tt arXiv:1103.3672}}.

\end{thebibliography}\endgroup
\end{document}